\documentclass[runningheads]{llncs}

\usepackage{amsfonts,epsfig}
\usepackage{graphicx}
\begin{document}

\date{}

\title{Polynomial kernels collapse the W-hierarchy}

\author{Liang Ding\inst{1*}, Abdul Samad\inst{1}, Xingran Xue\inst{1}, \\Xiuzhen Huang\inst{3}, and Liming Cai\inst{1,2}\thanks{To whom correspondence should be addressed. Email: cai@cs.uga.edu.}}
\institute{{\ \inst{1}Department of Computer Science, \ \inst{2}Institute of Bioinformatics \\
University of Georgia, Athens, GA 30602, USA\\
\email{lding@uga.edu, samad@uga.edu, xrxue@uga.edu, cai@cs.uga.edu}\\
\inst{3} Department of Computer Science\\
Arkansas State University, Jonesboro, AR 72467, USA}\\
\email{xhuang@astate.edu}}

\maketitle

\begin{abstract}
We prove that, for many parameterized problems in the class FPT, the existence of polynomial kernels implies the collapse of the W-hierarchy (i.e., W[P] = FPT). The collapsing results are also extended to assumed exponential kernels for problems in the class FPT. In particular, we establish a close relationship between polynomial (and exponential) kernelizability and the existence of sub-exponential time algorithms for a spectrum of circuit satisfiability problems in FPT. To the best of our knowledge, this is the first work that connects hardness for polynomial kernelizability of FPT problems to parameterized intractability. Our work also offers some new insights into the class FPT.
\end{abstract}


\section{Introduction}
Parameterized complexity theory studies computational behaviors of parameterized problems, for which the computational complexity is measured in the parameter values as well as the input size. The theory effectively characterizes the tractability of parameterized problems with dissociation between the parameter and the polynomial function in their time complexity. Therefore, it provides a viable approach to cope with computational intractability for problems with small parameter values. In particular, the class FPT has been established to capture all parameterized problems that can be solved in time $f(k) n^{O(1)}$, for any recursive function $f$ in the parameter value $k$ \cite{DowneyFellowsBook}. For two decades, a number of important techniques have been introduced to investigate the efficiency behaviors of parameterized problems in class FPT. Among them, kernelization has proved effective for the development of efficient parameterized algorithms \cite{BodlaenderSurvey2009,FellowsSurvey1,FellowsSurvey2,FellowsSurvey3}. 

A kernelization for a parameterized problem is to efficiently preprocess every problem instance whose size is reduced and thus bounded by a function in the parameter, preserving the answer. For example, $k$-{\sc Vertex Cover}, the problem of determining the existence of size $k$ vertex cover in a given graph can be reduced to the same problem on a subgraph (called kernel) of size at most $2k$ \cite{ChenVertexCover}. With this kernel size, the reduced problem becomes thus solvable in $O(2^{2k})$-time by exhaustive search on the kernel. Actually, the equivalence between parameterized-tractability and kernelizability can be proved without difficulty. Nevertheless, it is the kernel size that is of interest as it may fundamentally affects the time complexity needed by kernelization-based algorithms. Recent progresses have been made in pursuing polynomial size kernels for various parameterized problems \cite{LampisVertexCover,ThomasFVS,FominBidimensionality}. Most of the progresses have been generalized to the notion of meta-kernelization \cite{BodlaenderMetaKernel,GanianEtAl2013} that has been developed to extend linear kernelization algorithms to include larger classes of problems, e.g., through syntactic characterizations of problems.

On the other hand, limits for achieving desirable kernel sizes have also been extensively sought. For problems with some known kernel size upper bound, it is of interest to ask how likely it is to improve the upper bound.  In a few cases, such limitations were nicely built upon classical complexity hypotheses. For example, it is proved that a smaller kernel size than $2k$ would give rise to a breakthrough in designing approximation algorithms for the Minimum Vertex Cover problem \cite{ChenBoundKernel}. Recent progresses have particularly been made on problems that have resisted the trials of polynomial kernelizability. It has been shown that a larger number of problems in FPT, including $k$-{\sc Path} and $k$-{\sc Cycle}, do not admit polynomial kernels unless polynomial time hierarchy PH collapses to the third level \cite{BodlaenderNoPolyKernel,ChenLowerBoundKernel}. The work poses the open question whether polynomial kernels for these or other problems in FPT would cause a collapse of the W-hierarchy \cite{BodlaenderNoPolyKernel}. 

In this paper, we prove that, for a large number of problems in FPT, the existence of polynomial kernels or exponential kernels does collapse the W-hierarchy. To the best of our knowledge, this is the first such work to relate kernelization hardness of FPT problems to parameterized intractability. We establish the current work on the notion of {\it extended parameterized problems} introduced in \cite{Caisub-exponential}. There it was shown that subexponential-time algorithms for many circuit satisfiability problems $\pi$ in FPT would allow proofs of parameterized tractability for their extended versions $\pi^{\log n}$, which are W[P]-hard, thus collapsing the W-hierarchy. 
Here, we reveal additional relationships between such problems $\pi$ and their extended versions $\pi^{s(n)}$ (for many functions $s(n)=o(\log n)$) that are in the class FPT. In particular, we demonstrate there are subexponential-time transformations from $\pi$ to $\pi^{s(n)}$ that can reduce the parameter in any scale allowable by the function  $s(n)$. We then show that the existence of polynomial and even exponential kernels for problems $\pi^{s(n)}$ leads to subexponential-time algorithms for problems $\pi$ and thus collapse of the W-hierarchy. By the newly introduced transformation, we also provide a much simpler proof for the previously known conclusion \cite{Abrahamson1995,Chen2006lowerbound} that {\rm W[P]=FPT} implies the failure of the Exponential Time Hypothesis (ETH) \cite{Impagliazzo2001ETH,Impagliazzo2001kSAT}. Our work suggests that the difference between the two hypotheses is essentially the capability difference of problems SAT and 3SAT to admit $2^{o(k)}n^{O(1)}$-time algorithms, where $k$ is the number of boolean variables in the instances of the satisfiability problems.





\section{Preliminaries}\label{preSection}
Parameterized problems are computational problems whose inputs are of multi-variables, one of which is designated as the {\it parameter}. Let $\Sigma$ be the alphabet and $N$ be the natural number set. A parameterized problem is decision problem (or a language) with instances drawn from the product set $\Sigma^* \times N$. In particular, we use pairs $\langle I, k \rangle$ to represent instances of a parameterized problem, where $I \in \Sigma^*, k \in N$. We refer the reader to the book \cite{DowneyFellowsBook} for more technical definitions of parameterized problems and the theory of parameterized tractability.

The role of the parameter $k$ in the statement of a parameterized problem varies. Many such parameterized problems are formulated from optimization problems such that the problem is to answer how $\mathrm{OPT}(I)$, the optimum of solutions to the input instance $I$, is numerically related to the parameter \cite{CaiAndChen1997}. However, more sophisticated relationship may be formulated between a parameter $k$ and the problem instances. For this, we use the following general form for parameterized problem  definition. 

A parameterized problem $\pi$ is defined as the set of instances $\langle I, k \rangle$ that satisfy a predicate ${\cal P}_\pi$ that characterizes the problem $\pi$, i.e.,
\[ \pi = \{ \langle I, k \rangle \in \Sigma^* \times N:\, {\cal P}_\pi(I, k)\}\]
Typically parameterized problems formulated from NP problems can have existential second order predicates to characterize their instances. 

The parameterized problem format can be extended to formulate additional parameterized problems through ``engineering'' parameterizations \cite{Caisub-exponential}. In particular, for unbounded, nondecreasing function $s(n)$, where $n$ is the size of input instances, the following  parameterized problem $\pi^{s(n)}$ can be extended from problem $\pi$:
\[\pi^{s(n)} = \{ \langle I, k \rangle \in \Sigma \times N:\, {\cal P}_{\pi^{s(n)}}(I, k)\} \]

where ${\cal P}_{\pi^{s(n)}}(I, k) = {\cal P}_\pi(I, ks(n))$.  
This is because in problem $\pi^{s(n)}$, the parameter is still $k$ but the role of $k$ in the predicate  ${\cal P}_\pi$ is replaced by $ks(n)$. Extended parameterized problems have played an important role in connecting the existence of sub-exponential time algorithms to the collapse of the W-hierarchy.

\begin{figure}\label{circuitFigure}
\centering
\includegraphics[scale=0.29]{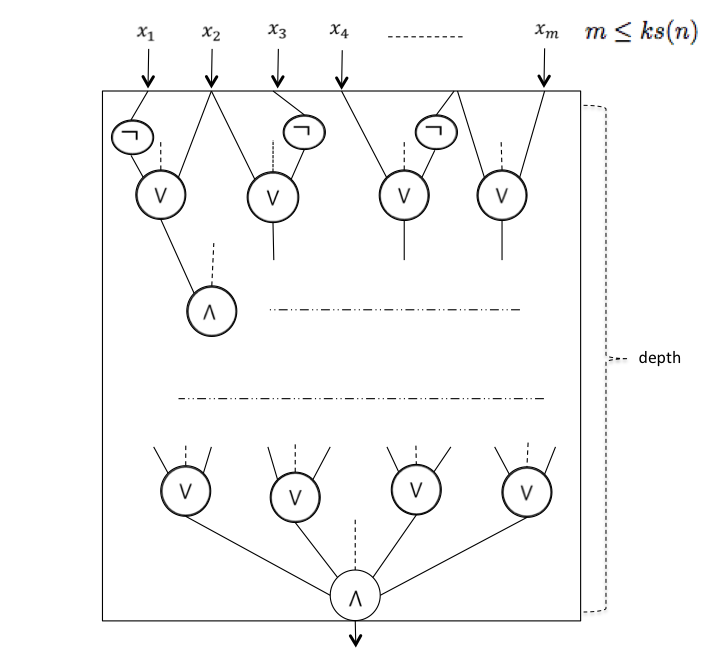}
\caption{Illustration of instance circuits for problem HwSAT$^{s(n)}$ in which there are $m$ input variables, $m\leq ks(n)$. AND ($\wedge$) and OR ($\vee$) gates are arranged in alternative levels, with negation ($\neg$) gates only appear as inputs to the first level gates and the output gate is designated as an AND gate. The problem seeks the answer to the question whether there is a satisfying assignment to the given circuit which sets exactly one-half of its input variables the value TRUE.}
\end{figure}

\begin{proposition}
\mbox{\sc \cite{Caisub-exponential}}\label{proposition_1}
If $\pi$ is solvable in time $O(2^{O(k)}p(n))$ for some polynomial $p$, then for any unbouned nondecreasing function $s(n)=o(\log n)$, $\pi^{s(n)}\in \mathrm{FPT}$.
\end{proposition}

The major technical proofs developed in \cite{Caisub-exponential} is centered on a series of circuit satisfiability problems defined on {\it normalized boolean circuits}.  Such circuits arrange gates in levels, alternating AND and OR gates such that gates at level $i$ receive inputs only from gates at level $i-1$ and have an AND gate as their output gate. With such circuits, we define parameter problem 
\[\mathrm{HwSAT} =\{ \langle C, k \rangle : {\cal P}_{\rm {HwSAT}}(C, k)\}\]

where the predicate ${\cal P}_{\mathrm{HwSAT}}(X, y)$ reads as: normalized boolean circuit $X$ of at most $y$ input variables is satisfied with some assignment that sets exactly one half of its input variables to TRUE. The extended version of the problem, for any unbounded, nondecreasing function $s(n)$, is thus:

\[\mbox{HwSAT}^{s(n)} =\{ \langle C, k \rangle : {\cal P}_{\rm HwSAT}(C, ks(n))\}\]
where circuit $C$ has at most $ks(n)$ input variables, where $n=|\langle C, k\rangle|$. Figure 1 illustrates circuit instances for problem {\sc HwSAT}$^{s(n)}$.

\begin{proposition}\mbox{\sc \cite{Caisub-exponential}}\label{proposition_2}
Let $s(n)=\Omega(\log n)$ be a unbounded, nondecreasing function. Then $\mathrm{HwSAT}^{s(n)}$ is $\mathrm{W[P]}$-hard.
\end{proposition}


\begin{proposition}\mbox{\sc \cite{Caisub-exponential}}\label{proposition_3}
For any nondecreasing function $s(n)=o(\log n)$, if problem {\rm HwSAT} can be solved in $O(2^{o(k)} n^{O(1)})$-time, then $\mathrm{HwSAT}^{\log n}$ is in $\mathrm{FPT}$ and thus $\mathrm{W[P]=FPT}$.
\end{proposition}

The work in \cite{Caisub-exponential} also defined, for every integer $i\geq 1$,  satisfiability problems on normalized circuits of depth $i$ 
\[ \mbox{HwSAT}[i] = \{ \langle C, k \rangle : {\cal P}_{\rm {HwSAT}[i]}(C, k)\}\]
and their extended versions HwSAT$[i]^{s(n)}$ for any unbounded, nondecreasing function $s(n)$. For every $i\geq 1$, analogs of Propositions \ref{proposition_1}-\ref{proposition_3} were proved for HwSAT$[i]^{s(n)}$ with respect to W$[i]$, the $i$-th level of the W-hierarchy.

In this paper, we will focus our discussions on problems HwSAT$^{s(n)}$, for many nondecreasing functions $s(n)=o(\log n)$,  and show that they do not admit polynomial even exponential kernels by relating such kernelization to sub-exponential time algorithms for HwSAT. Our techniques also allow us to achieve analogous results for HwSAT$[i]$ and HwSAT$[i]^{s(n)}$, $i\geq 1$, whose proofs will be omitted due to the space limitation.


\vspace{1mm}
A kernelization of a parameterized problem $\pi$ is a polynomial time (in both $n$ and $k$) reduction that reduces any instance $\langle I, k\rangle$ of $\pi$ to an  instance $\langle I', k'\rangle$ of the same problem such that
\begin{enumerate}
\item ${\cal P}_\pi(I, k)$ if and only if ${\cal P}_\pi(I', k')$; and
\item $|I'| \leq f(k)$, and $k' \leq f(k)$. 
\end{enumerate}
where $f$ is a computable function and $f(k)$ is called the {\it kernel size}. Problem $\pi$ admits a {\it polynomial} (resp. {\it exponential}) {\it kernel} if the function $f$ in the above definition of kernelization is a polynomial (resp. exponential) function. Recently it has been shown that polynomial kernels may not exist for many parameterized problems under classical complexity hypotheses \cite{BodlaenderNoPolyKernel,ChenLowerBoundKernel}. Connections between the existence of such kernels to the structural properties of parameterized complexity, in particular the W-hierarchy, has also been sought \cite{BodlaenderNoPolyKernel}. In the rest of this paper, we demonstrate such a connection. We will prove that parameterized problems HwSAT$^{s(n)}$ in FPT, for many functions $s(n) = o(\log n)$, do not have polynomial kernel, unless problem HwSAT can be solved by subexponential time algorithms and the W-hierarchy collapses.


\section{Subexponential-Time Transformations}\label{tranSection}
 In this section, we first introduce a novel technique that is critical to achieving the results in this paper. The technique transforms problems HwSAT to HwSAT$^{s(n)}$ such that the  parameter values can be reduced by any 
 scale permitted by the function $s(n)$.

In particular, to transform an instance $\langle C, k\rangle$ for HwSAT$^{s(n)}$, an instance $\langle C_1, k_1\rangle$ for HwSAT$^{s(n)}$, the basic idea is to rearrange the $k$ input variables of circuit $C$ into no more than $k_1 \times s(n_1)$ input variables in the constructed circuit $C_1$ (where $n_1 =|\langle C_1, k_1\rangle|$), while preserving the circuit logic of $C$ in $C_1$. The new parameter $k_1$ can be chosen as desired.

\begin{lemma}\label{subTransLemma}
Let $s(n) = \log n / t(n)$ be a function for some non-decreasing unbounded function $t(n)=\omega(1)$.  
Then there is a reduction $f$ such that, for every instance $\langle C, k\rangle$,
\begin{enumerate}
\item $f(\langle C, k\rangle) = \langle C_1, k_1\rangle$, and
\item ${\cal P}_{\rm HwSAT} (C, k) \Longleftrightarrow  {\cal P}_{\rm HwSAT}(C_1, k_1s(n_1))$, and 
\item $f$ is computable in time $2^{o(k)} n^{O(1)}$,
\end{enumerate}
where $n=|\langle C, k\rangle|$, $n_1=|\langle C_1, k_1\rangle|$, and $k_1 =l(k)$, for some function $l(k)\geq r(k)t(2^k)$ and some function $r(k)=\omega(1)$.
\end{lemma}

\begin{proof}
We demonstrate how instance $\langle C_1, k_1\rangle$ for HwSAT$^{s(n)}$ is constructed from given instance $\langle C, k\rangle$ for HwSAT. We will only consider the case $2^{k} \geq n = |\langle C, k\rangle|$, because for the other case $2^{k} \leq n$ we can directly answer for the instance $\langle C, k\rangle$ with an exhaustive search that runs in time $2^{O(k)} n^{O(1)} = n^{O(1)}$, a polynomial time.

So we assume that $2^{k} \geq n$ in the rest of the proof. 
The circuit  $C_1$ is defined as
\[ C_1 = \left\{ 
  \begin{array}{l l}
    C & \quad \mbox{ if } n \geq  n_0\\
    C \#0^{n_0 - n} & \quad \mbox{ if } n < n_0 \\
  \end{array} \right.
\]
where $n_0$ is chosen to be at least $2^{k/r(k)}$, for some nondecreasing $r(k)=\omega(1)$.
It is clear that the circuit logic of $C_1$ remains the same as $C$.
\begin{claim}
Let $k_1 = l(k)$ be chosen as $r(k)t(2^k)$. Then the constructed circuit $C_1$ contains at most $k_1 s(n_1)$ input variables, where $n_1 = |\langle C_1, k_1\rangle|$.
\end{claim}
\begin{proof} 
It suffices to show that
\[k_1 s(n_1) \geq k \, \, \mbox {or simply } \, s(n_1) \geq \frac{k}{l(k)} \]

Because the size of $C_1$ is at least $n_0 = 2^{k/r(k)}$, we have $n_1 \geq 2^{k/r(k)} $. So
\begin{eqnarray}
s(n_1) = \frac{\log n_1}{t(n_1)} \geq \frac{\log 2^{k/r(k)}}{t(n_1)} = \frac{k}{r(k)t(n_1)}
\end{eqnarray}\label{lemmaEqu1}

On the other hand, since $C_1$ is either $C$ or $C$ padded with dummy gates up to the size $n_0$, $|C_1| \leq \max\{ |C|, n_0\}$. By the definition of $k_1$, $k_1 \leq k$. Also because of $n \leq 2^k$, we have $n_1 = |\langle C_1, k_1\rangle| \leq \max\{n, n_0\} \leq 2^k$. Replacing it in the above inequality,

\begin{eqnarray}
s(n_1)&\geq &\frac{k}{r(k)t(n_1)}\geq \frac{k}{r(k) t(2^k)} \nonumber\\
&=& \frac{k}{l(k)}\cdot \frac{l(k)}{r(k) t(2^k)} \geq \frac{k}{l(k)}.\label{lemmaEqu2}
\end{eqnarray}
The last inequality in (\ref{lemmaEqu2}) is ensured by the chosen  
$l(k) \geq r(k) t(2^k)$.
(End of proof for the claim).
\end{proof} 


It is clear that the constructed circuit $C_1$ has $k$ input variables which is bounded by $k_1  s(n_1)$. Moreover, the transformation function $f$ satisfies  
\[ {\cal P}_{\rm HwSAT}(C, k) \Longleftrightarrow {\cal P}_{\rm HwSAT} (C_1, k_1s(n_1)) \] 
as the circuit logic of $C$ remains the same in  $C_1$. Finally,  
the transformation (including the padding) can be done in the subexponential time $2^{k/r(k)} n^{O(1)}$. \qed
\end{proof}

\newpage

Note that the Lemma \ref{subTransLemma} holds without specifying the function $r(n)$. To use the subexponential-time transformation to produce instances $\langle C_1, k_1\rangle$ with scaled down parameter $k_1$ for the interest of this paper, we tailor the lemma with a number of concrete functions $s(n)$ to specify $r(k)$ and $l(k)$. In general, for each function $t(n)=o(h(n))$, we can always choose some $l(k)$ in the class $o(h(2^k))$. This is because we can set $r(k)$ to be such that $t(n)=h(n)/r(n)^2$, where $r(n) = \omega(1)$.

Table \ref{transLemmaTable} lists three categories of functions $s(n)$ and the classes from which the corresponding function $l(k)$ belongs to. The most restricted category contain functions $s(n) = \Omega(\log^{\epsilon} n)$ for any given constant $\epsilon>0$, for which $t(n) = o(\log^{1-\epsilon} n)$;  the transformation can reduce the parameter from $k$ to some $l(k)=o(k^{\frac{1}{d}})$, for any $d < 1/(1-\epsilon)$. The second category contains functions $s(n)$ for which $t(n) = o(\log^{\epsilon} n)$ for all $\epsilon>0$, allowing $l(k)=o(k^{\frac{1}{d}})$, for any $d\geq 1$ independent of $\epsilon$. The third category considers functions $s(n)=\log n/o(\log \log n)$  which allows further reduction of the parameter to $l(k)=o(\log k)$ and thus $O(\log k)$.


\begin{table}\label{transLemmaTable}
\begin{center}
    \begin{tabular}{| c | c | c | c |}
    \hline
     Category & $t(n)$ & corresponding class for $l(k)$ \\ \hline\hline 
    1  &  $o(\log^{1-\epsilon} n)$  for any given $\epsilon>0$ &  $o(k^{\frac{1}{d}})$ for $d<1/(1-\epsilon)$\\ \hline

     2 & $o(\log^\epsilon n)$ for all $\epsilon >0$ & $o(k^{\frac{1}{d}})$ for any given $d\geq 1$\\ \hline
     3 & $o({\log \log n})$ & $O(\log k)$ \\ \hline
     
       \hline
    \end{tabular}
\end{center}
\caption{Three categories of functions $s(n)$ for which Lemma \ref{subTransLemma} holds and corresponding classes from which $l(k)$ can be picked to be of interest to this paper. $s(n) = \log n/t(n)$.} 
\end{table}

\section{Polynomial Kernels Collapse the W-Hierachy}\label{polySection}
In this section, we will show that polynomial kernels for HwSAT$^{s(n)}$, for many functions $s(n)=o(\log n)$, would allow subexponential-time algorithms to solve HwSAT. According to \cite{Caisub-exponential}, such an algorithm would collapse the W-hierarchy. We present details in the following.

\begin{lemma}\label{polyRuntimeLemma}
Let $d\geq 1$ and $s(n)=o(\log n)$ be an bounded, non-decreasing function. If $\mathrm{HwSAT}^{s(n)}$ admits polynomial kernels of size bounded by $k^d$, then it is solvable in time $2^{O(k^d)}k^d + n^{O(1)}$, where $k$ is the parameter of problem $\mathrm{HwSAT}^{s(n)}$.
\end{lemma}

\begin{proof}
Each input to problem $\mathrm{HwSAT}^{s(n)}$ is of the form $\langle C, k\rangle$, in which $C$ is a normalized boolean circuit with $ks(n)$ input variables, where $n=|\langle C, k \rangle|$. 
The assumed polynomial kernel reduces instance $\langle C, k\rangle$ to $\langle C', k'\rangle$, for which $|C'| \leq k^d$ and $k' \leq k^d$, such that
\[{\cal P}_{\rm HwSAT} (C, ks(n)) \Longleftrightarrow  {\cal P}_{\rm HwSAT}(C', k's(n'))\]
where $n' = |\langle C', k'\rangle|$. Because $C'$ is of a size bounded by $k^d$, the number of input variable in $C'$ cannot be more than this number. So the answer to the question on instance $\langle C', k'\rangle$ can be found by enumerating all half-weight assignments to the $ k^d$ variables in time $2^{O(k^d)} k^d$. Together with the time used by the kernelization, the total time needed to solve problem $\mathrm{HwSAT}^{s(n)}$ is $2^{O(k^d)}k^d + n^{O(1)}$. \qed
\end{proof}

\begin{theorem}\label{polyCollapseTheorem}
Let $d\geq 1$ be any given integer and $s(n) = \log n / t(n)$ be a function for some non-decreasing unbounded function $t(n)=\omega(1)$ and $t(n) = o(\log^{\frac{1}{d}} n)$. Then unless $\mathrm{W[P]=FPT}$, problem $\mathrm{HwSAT}^{s(n)}$ cannot admit polynomial kernels of size $k^d$.
\end{theorem}

\begin{proof}
We show that the assumed polynomial kernels for HwSAT$^{s(n)}$ leads to a subexponential-time algorithm for HwSAT, which collapses the W-hierarchy by Proposition \ref{proposition_3} \cite{Caisub-exponential}. We sketch the steps in the following.

First, for $t(n)=o(\log^{\frac{1}{d}} n)$, Lemma \ref{subTransLemma} holds for the first and second categories in Table \ref{transLemmaTable}, which transform the given instance $\langle C, k\rangle$ for HwSAT to an instance $\langle C_1, k_1\rangle$ for HwSAT$^{s(n)}$ in time $O(2^{o(k)} n^{O(1)})$, where $k_1 = o(k^{\frac{1}{d}})$.

Second, consider the transformed instance $\langle C_1, k_1\rangle$ in which circuit $C_1$ has at most $k_1 s(n_1)$ input variables, where $n_1=|\langle C_1, k_1\rangle|$. With the polynomial kernelization, another instance $\langle C', k'\rangle$ is produced, such that 
\begin{enumerate}
\item ${\cal P}_{\rm HwSAT}(C_1, k_1s(n_1) \Longleftrightarrow {\cal P}_{\rm HwSAT} (C', k's(n')$, where $n'=|\langle C', k'\rangle|$, 
\item $|C'| \leq k_1^d$, and $k' \leq k_1^d$.
\end{enumerate}

Third, by Lemma \ref{polyRuntimeLemma}, we can enumerate and check all possible assignment combinations to the input variables to determine if the predicate ${\cal P}_{\mathrm{HwSAT}}(C', k'$ $s(n'))$ is true and thus to determine if the predicate ${\cal P}_{\rm HwSAT}(C_1, k_1s(n_1))$ is true, and consequentially if ${\cal P}_{\rm HwSAT}(C, ks(n))$ is true, within time
\begin{eqnarray}
 2^{O(k')} |C'| + n_1^{O(1)} &\leq &2^{O(k_1^d)} k_1^d + n_1^{O(1)} \nonumber\\
 &\leq & 2^{o(k^{\frac{1}{d}})^d} (k^{\frac{1}{d}})^d + n_1^{O(1)} \nonumber\\
 &=& 2^{o(k)} k + n_1^{O(1)}\nonumber\\
 &\leq & 2^{o(k)} k + \max\{n^{O(1)}, (2^{o(k)})^{O(1)}\}\nonumber\\
 &\leq & 2^{o(k)} k + n^{O(1)}\nonumber
\end{eqnarray}
Therefore, problem HwSAT admits subexponential-time algorithms. By Proposition \ref{proposition_3}, W[P] = FPT. \qed
\end{proof}


By the proof of Theorem \ref{polyCollapseTheorem}, and  categories 1 and 2 in Table \ref{transLemmaTable}, we have

\begin{corollary}
Let $\epsilon$ be a fixed number, $0<\epsilon < 1$, and function $s(n)=\log^\epsilon n$. Then unless {\rm W[P] = FPT}, problem $\mathrm{HwSAT}^{s(n)}$ does not admit polynomial kernels of size $k^d$ for any $d <1/(1-\epsilon)$.
\end{corollary}

\begin{corollary}
Let  unbounded, non-decreasing function $s(n) = \log n / t(n)$, where $t(n)=o(\log^\epsilon n)$, for all $\epsilon >0$. Then unless {\rm W[P] = FPT}, problem $\mathrm{HwSAT}^{s(n)}$ does not admit polynomial kernels. \end{corollary}

Before concluding this section, we point out that using the same techniques, hardness for polynomial kernelizability can be proved for some other classes of problems in FPT. In particular, according to the work of \cite{Caisub-exponential}, half-weighted satisfiability problems were also defined for normalized boolean circuits of constant depths. In the same spirit of HwSAT, for every $t\geq 1$, problem HwSAT$[i]$ has been defined to determines if a given depth-$i$ normalized boolean circuit of $k$ input variables is satisfiable by some half-weighted assignment. So we are able to draw the following conclusions. However, due to the page limitation, we omit their proofs. 

\begin{theorem}
Let $d\geq 1$ be any given integer and $s(n) = \log n / t(n)$ be a function for some non-decreasing unbounded function $t(n)=\omega(1)$ and $t(n) = o(\log^{\frac{1}{d}} n)$. Then for every $i\geq 1$, unless $\mathrm{W[}i\mathrm{]=FPT}$, problem $\mathrm{HwSAT}[i]^{s(n)}$ cannot admit polynomial kernels of size $k^d$.
\end{theorem}

\begin{corollary}
Let $\epsilon$ be a fixed number, $0<\epsilon< 1$, and function $s(n)=\log^\epsilon n$. Then for every $i\geq 1$, unless {\rm W[$i$] = FPT}, problem $\mathrm{HwSAT}[i]^{s(n)}$ does not admit polynomial kernels of size $k^d$ for any $d <1/(1-\epsilon)$.
\end{corollary}

\begin{corollary}
Let unbounded, non-decreasing function\ \ $s(n)=\log n/t(n)$, where $t(n)=o(\log^\epsilon n)$, for all $\epsilon >0$. Then for every $i\geq 1$, unless {\rm W[$i$] = FPT}, problem $\mathrm{HwSAT}[i]^{s(n)}$ does not admit polynomial kernels. \end{corollary}

\section{Exponential Kernels Collapse the W-Hierachy}\label{subexSextion}
In this section, we extend the techniques in Section \ref{polySection} to show that exponential kernels for HwSAT$^{s(n)}$ for many functions $s(n)$ would also collapse the W-hierarchy. 

\begin{lemma}\label{expRuntimeLemma}
Let $s(n)=o(\log n)$ be an bounded, non-decreasing function. If problem $\mathrm{HwSAT}^{s(n)}$ admits exponential kernels of size bounded by $2^{O(k)}$, then it is solvable in time $2^{2^{O(k)}}2^{O(k)} + n^{O(1)}$, where $k$ is the parameter of problem $\mathrm{HwSAT}^{s(n)}$.
\end{lemma}

\begin{proof}
Similar to the proof for Lemma \ref{polyRuntimeLemma}, the 
assumed exponential kernelization reduces instance $\langle C, k\rangle$ to $\langle C', k'\rangle$, for which $|C'| \leq 2^{O(k)}$ and $k' \leq 2^{O(k)}$, such that
\[{\cal P}_{\rm HwSAT} (C, ks(n)) \Longleftrightarrow  {\cal P}_{\rm HwSAT}(C', k's(n'))\]
where $n=|\langle C, k \rangle|$ and $n' = |\langle C', k'\rangle|$. Clearly the number of input variables in $C'$ is bounded by $2^{O(k)}$; the answer to the question on instance $\langle C', k'\rangle$ can be found by enumerating all half-weight assignments to the $ 2^{O(k)}$ variables in time $2^{2^{O(k)}} 2^{O(k)}$. And the total time needed to solve problem HwSAT$^{s(n)}$ is $2^{2^{O(k)}}2^{O(k)} + n^{O(1)}$. \qed
\end{proof}

\begin{theorem}\label{expoCollapseTheorem}
Let$s(n) = \log n / t(n)$ be a function for some non-decreasing unbounded function $t(n)=\omega(1)$ and $t(n) = o(\log\log n)$. Then unless $\mathrm{W[P]=FPT}$, problem $\mathrm{HwSAT}^{s(n)}$ cannot admit exponential kernels of size $2^{O(k)}$.
\end{theorem}

\begin{proof}
We show that the assumed exponential kernels for HwSAT$^{s(n)}$ leads to a subexponential-time algorithm for HwSAT, which collapses the W-hierarchy by Proposition \ref{proposition_3}. We sketch the steps in the following.

First, without loss of generality, we assume HwSAT$^{s(n)}$ has kernel of size $2^{Mk}$, where $M\geq 0$ is a constant. Then category 3 in Table \ref{transLemmaTable} is applied to transform the given instance $\langle C, k\rangle$ for HwSAT to produce instance $\langle C_1, k_1\rangle$ for HwSAT$^{s(n)}$ in time $O(2^{o(k)} n^{O(1)})$, where $k_1 = \frac{1}{M+\delta}\log k$ for some $\delta>0$.

Second, consider the transformed instance $\langle C_1, k_1\rangle$ in which circuit $C_1$ has at most $k_1 s(n_1)$ input variables, where $n_1=|\langle C_1, k_1\rangle|$. With the exponential kernelization, another instance $\langle C', k'\rangle$ is produced, such that 
\begin{enumerate}
\item ${\cal P}_{\rm HwSAT}(C_1, k_1s(n_1) \Longleftrightarrow {\cal P}_{\rm HwSAT} (C', k's(n')$, where $n'=|\langle C', k'\rangle|$, 
\item $|C'| \leq 2^{Mk_1}$, and $k' \leq 2^{Mk_1}$.
\end{enumerate}

Third, by Lemma \ref{polyRuntimeLemma}, we can enumerate and check all possible assignment combinations to the input variables to determine if the predicate ${\cal P}_{\rm HwSAT}(C', k'$ $s(n'))$ is true and thus to determine if the predicate ${\cal P}_{\rm HwSAT}(C_1, k_1s(n_1))$ is true, and consequentially if ${\cal P}_{\rm HwSAT}(C, ks(n))$ is true, within time
\begin{eqnarray}
 2^{k'} |C'| + n_1^{O(1)}  &\leq &  2^{2^{Mk_1}} 2^{Mk_1} + n_1^{O(1)} \nonumber\\
 &\leq &  2^{2^{M\cdot\frac{1}{M+\delta}\log k}} 2^{M\cdot\frac{1}{M+\delta}\log k} + n_1^{O(1)} \nonumber\\
 &=& 2^{k^{\frac{M}{M+\delta}}} k^{\frac{M}{M+\epsilon}} + n_1^{O(1)}\nonumber\\
 &\leq & 2^{k^{\frac{M}{M+\delta}}} k + \max\{n^{O(1)}, (2^{o(k)})^{O(1)}\}\nonumber\\
 &\leq & 2^{k^{\frac{M}{M+\delta}}} k + n^{O(1)}\nonumber
\end{eqnarray}
Therefore, problem HwSAT admits subexponential-time algorithms. By Proposition \ref{proposition_3}, W[P] = FPT. \qed
\end{proof}

\section{Further Consequences of the Transformation Lemma}\label{furtherResultsSection}

In this section, we show that Lemma \ref{subTransLemma} allows further revelations of the rich parameterized complexities for within the class FPT. Through a much simpler proof for the previously known conclusion \cite{Abrahamson1995,Chen2006lowerbound} that {\rm W[P]=FPT} implies the failure of the Exponential Time Hypothesis (ETH) \cite{Impagliazzo2001ETH,Impagliazzo2001kSAT}, our work suggests that the difference between the two hypotheses is essentially the capability difference of problems SAT and 3SAT to admit $2^{o(k)}n^{O(1)}$-time algorithms, where $k$ is the number of boolean variables in the instances of the satisfiability problems.

\begin{theorem}\label{HwSAT_admits_subexp_theorem}
{\rm W[P]=FPT} implies {\rm HwSAT} admits subexponential-time algorithms.
\end{theorem}

\begin{proof}
Assume W[P]=FPT and $\mathrm{HwSAT}^{\log n}$ is solvable in $f(k)n^{O(1)}$, where $f(k)=\omega(2^k)$ and its inverse funciton is computable in polynomial time. Let $s(n)=\log n/t(n)$ with $t(n)=o(f^{-1}(\log n))$. By Lemma \ref{subTransLemma}, we choose $k_1=o(t(2^k))=o(f^{-1}(k))$ in the transformation from the given instance $\langle C, k \rangle$ of HwSAT to an instance $\langle C_1, k_1 \rangle$ of HwSAT$^{s(n)}$. 

Since HwSAT$^{\log n}\in \mathrm{FPT}$, it has kernel size $f(k_1)$. Then HwSAT$^{s(n)}$ has kernel size at most $f(k_1)$. We apply the kernelization to produce another instance $\langle C', k' \rangle$. Similar to Theorem \ref{expoCollapseTheorem}, by enumerating and checking all possible assignment combinations to the input variables of $\langle C', k' \rangle$, problem HwSAT can be solved in: 
\begin{eqnarray}
 2^{k'} |C'| + n_1^{O(1)}  &\leq &  2^{f(k_1)} f(k_1) + n_1^{O(1)} \nonumber\\
 &\leq &  2^{f(o(f^{-1}(k))} f(f^{-1}(k)) + n_1^{O(1)} \nonumber\\
 &\leq & 2^{o(k)} k + n^{O(1)}\nonumber
\end{eqnarray}
Therefore, problem HwSAT admits subexponential time algorithms.
\qed
\end{proof}

%

It has been shown that the SAT problem, when the number of variables is the designated parameter, can be turned into HwSAT \cite{Caisub-exponential}.
Thus by Theorem\ref{HwSAT_admits_subexp_theorem} and Proposition~\ref{proposition_3}, we conclude


\begin{corollary}
{\rm W[P]=FPT} implies {\rm 3SAT} admits 
subexponential-time algorithms {\rm (}i.e., the failure of {\rm ETH}{\rm )}.
\end{corollary}

\begin{corollary}
{\rm W[P]=FPT} if and only if {\rm HwSAT} admits subexponential-time algorithms.
\end{corollary}


In the following theorem, we genelize the results about the nonexistence of polynomial kernels and exponential kernels of $\mathrm{HwSAT}^{s(n)}$ to any polynomial time invertible function $g(k)$.

\begin{theorem}\label{general}
Unless {\rm W[P]=FPT}, for any given function $g(k)$ whose inverse function is computable in polynomial time, there exists infinite functions $h(k)$, such that problem $\mathrm{HwSAT}^{s(n)}$, for $s(n)=\frac{\log n}{(g^{-1}(\log n))/h(n)}$,
does not admit kernel of size less than or equal to $g(k)$.
\end{theorem}

\begin{proof}
According to Lemma \ref{subTransLemma} and by similar technique in Theorem \ref{HwSAT_admits_subexp_theorem}, we only need to make sure $k_1=o(g^{-1}(k))$. By Lemma \ref{subTransLemma}, $t(k)=\frac{g^{-1}(\log k)}{h(k)}$ and $k_1=l(k)\geq r(k)t(2^k)$, where $r(k)=\omega(1)$. So we have $t(2^k)=\frac{g^{-1}(k)}{h(2^k)}$. 
Thus any function $h(k)$ satisfying $h(2^k)=o(g^{-1}(k))$ and $h(2^k)=\omega(1)$ meets our needs.
\qed
\end{proof}

\section{Conclusion}\label{discussionConcluSection}
We have proved that many parameterized problems in FPT do not admit polynomial kernels (or even exponential kernels) unless the W-hierarchy collapses. To the best of our knowledge, this is the first such result connecting difficulty for polynomial kernelization to parameterized intractability. The techniques we developed are novel, which establish a close relationship between polynomial (and exponential) kernelizability and the existence of sub-exponential algorithms for a spectrum of circuit satisfiability problems in FPT.

According to Theorem~\ref{general}, for every suitable function $s(n)$, unless W[P] =FPT, problem HwSAT$^{s(n)}$ cannot admit kernels of size at most $g(k)$, where $s(n)=\frac{\log n}{(g^{-1}(\log n))/h(n)}$ for any suitable $h$. It is natural to ask if the (conditional) kernel size lower bounds achieved for the problem can be further improved. We point out that the same problem has a kernel size upper bound $2^{ks(n)}$ (with the assumption $n <2^{ks(n)}$, the time used to enumerating all assignments to the circuit input variables), which is bounded by $2^{kg(k)}$. Techniques that could narrow down the exponential gap between the lower and upper bounds may also tell if polynomial kernels for HwSAT problem would lead to the collapsing of the W-hierarchy.

Another interesting issue raised is how the non-kernelizability developed by our work for HwSAT$^{s(n)}$ problems be related to kernelizability hardness results for many problems like $k$-{\sc Simple Path} and $k$-{\sc Cycle} recently achieved by others \cite{BodlaenderNoPolyKernel}. Can such ``more natural'' problems  be transformed from HwSAT$^{s(n)}$ or HwSAT$[i]^{s(n)}$ (for some $i\geq 1$) through some sort of kernel size-preserving reduction,  for some function $s(n)=o(\log n)$ and $s(n)=\omega(1)$?

\section*{Acknowledgments}
This work was supported in part by a research grant from the National Science Foundation of the U.S.A. (award No: 0916250).

\end{document}